\documentclass[journal]{IEEEtran}

\usepackage[utf8]{inputenc} 
\usepackage[T1]{fontenc}    
\usepackage{hyperref}       
\usepackage{url}            
\usepackage{booktabs}       
\usepackage{amsfonts}       
\usepackage{nicefrac}       
\usepackage{microtype}      

\usepackage{graphicx}
\usepackage{mathrsfs}
\usepackage{algorithmicx}
\usepackage{algorithm,algpseudocode,float}
\usepackage{amssymb}
\usepackage{amsmath}
\usepackage{float}
\usepackage{subfig}
\usepackage{amsthm}
\usepackage{cite}
\newtheorem{lemma}{Lemma}
\newtheorem{theorem}{Theorem}
\newtheorem{proposition}{Proposition}

\algdef{SE}[FOR]{NoDoFor}{EndFor}[1]{\algorithmicfor\ #1}{\algorithmicend\ \algorithmicfor}%
\algdef{SE}[IF]{NoThenIf}{EndIf}[1]{\algorithmicif\ #1}{\algorithmicend\ \algorithmicif}%
\algdef{SE}[WHILE]{NoDoWhile}{EndWhile}[1]{\algorithmicwhile\ #1}{\algorithmicend\ \algorithmicwhile}%
\makeatletter
\newenvironment{breakablealgorithm}
  {
   \begin{center}
     \refstepcounter{algorithm}
     \hrule height.8pt depth0pt \kern2pt
     \renewcommand{\caption}[2][\relax]{
       {\raggedright\textbf{\ALG@name~\thealgorithm} ##2\par}%
       \ifx\relax##1\relax 
         \addcontentsline{loa}{algorithm}{\protect\numberline{\thealgorithm}##2}%
       \else 
         \addcontentsline{loa}{algorithm}{\protect\numberline{\thealgorithm}##1}%
       \fi
       \kern2pt\hrule\kern2pt
     }
  }{
     \kern2pt\hrule\relax
   \end{center}
  }
\makeatother

\begin{document}
%
\title{Fast Budgeted Influence Maximization over Multi-Action Event Logs}

%
%
%

\author{Qilian~Yu,~\IEEEmembership{Student Member,~IEEE,}
        Hang~Li,~\IEEEmembership{Member,~IEEE,}\\
        Yun~Liao,~\IEEEmembership{Student Member,~IEEE,}
        and~Shuguang~Cui,~\IEEEmembership{Fellow,~IEEE}
\thanks{Q. Yu, H. Li and S. Cui are with the Department
of Electrical and Computer Engineering, University of California, Davis, CA, 95616 USA.
 Emails: \{dryu, hdli, sgcui\}@ucdavis.edu.}
\thanks{Y. Liao is with the Department
of Electrical and Computer Engineering, University of California, San Diego, CA, 92093 USA. Email: yunliao@ucsd.edu.}
\thanks{The short conference version of this paper has been submitted to ICASSP 2018.}
}

\maketitle

\begin{abstract}
In a social network, influence maximization is the problem of identifying a set of users that own the maximum {\it influence ability} across the network. In this paper, a novel credit distribution (CD) based model, termed as the multi-action CD~(mCD) model, is introduced to quantify the influence ability of each user, which works with practical datasets where one type of action could be recorded for multiple times. Based on this model, influence maximization is formulated as a submodular maximization problem under a general knapsack constraint, which is NP-hard. An efficient streaming algorithm with one-round scan over the user set is developed to find a suboptimal solution. Specifically, we first solve a special case of knapsack constraints, i.e., a cardinality constraint, and show that the developed streaming algorithm can achieve ($\frac{1}{2}-\epsilon$)-approximation of the optimality. Furthermore, for the general knapsack case, we show that the modified streaming algorithm can achieve ($\frac{1}{3}-\epsilon$)-approximation of the optimality.  Finally, experiments are conducted over real Twitter dataset and demonstrate that the mCD model enjoys high accuracy compared to the conventional CD model in estimating the total number of people who get influenced in a social network. Moreover, through the comparison to the conventional CD, non-CD models, and the mCD model with the greedy algorithm on the performance of the influence maximization problem, we show the effectiveness and efficiency of the proposed mCD model with the streaming algorithm.
\end{abstract}

\begin{IEEEkeywords}
Online Social Networks, Influence Maximization, Credit Distribution, Streaming Algorithm.
\end{IEEEkeywords}

%
\IEEEpeerreviewmaketitle

\section{Introduction}\label{sec:introduction}
As  information technology advances, people get informed through numerous media channels spanning over conventional media (e.g., newspapers, radio or TV programs) and modern social media (e.g., mobile APPs, electronic publications, or world wide web (WWW)). Since computers and smart phones become more and more popular, information now spreads at a speed faster than ever before. In particular, people are ubiquitously connected by online social networks, and a person's behavior may quickly affect others, who may further perform some relevant actions. For example, after a celebrity posts a new message on Twitter, many followers read this tweet and then retweet. Then, the friends of these followers may do such actions as well. Consequently, the same tweet could be posted again and again, while more and more people are involved. This phenomenon in social networks is referred to as {\it influence propagation}. Here, such a celebrity could be called the {\it influencer}. Note that, in general, there may be more than one influencer for one particular event. 

It is easy to see that influencers may have significant impacts on the dynamics in social networks, and thus the problem of influencer identification has drawn great attention in both academia and industry~\cite{1,2}.~One such pioneer work is about viral marketing~\cite{3,4}, where a new product is advertised to a target group of people such that the advertisement could reach a large fraction of the social network users. In Later work~\cite{5}, the influencer identification problem is commonly formulated as an \textbf{influence maximization problem}: Given an influence propagation model, find $k$ ``seed'' nodes such that the expected number of nodes that eventually get ``influenced'' is maximized. Two propagation models, i.e., the Independent Cascade (IC) model and the Linear Threshold (LT) model, are widely used. In these two graph based models, one of the most important parameters is the edge weight, which represents the probability that a person gets influenced and takes a similar action as what his or her socially connected neighbors do. In existing works~\cite{5,6,7,8,9}, the weight of each edge is usually determined by one of the following methods: 1) assigning a constant (e.g., 0.01), 2) assigning a random number from a small set (e.g., \{0.1, 0.01, 0.001\}), 3) assigning the reciprocal of a node's in-degree, or 4) assigning a value based on real data.

Although accelerated greedy algorithms have been developed~\cite{10,11} to mitigate the high computation cost in influence maximization, all works mentioned above~\cite{5,6,7,8,9,10,11} need significant Monte-Carlo simulations to calculate the expected number of influenced nodes, which prevents their results from being implemented in analyzing large-scale social networks. Recently, a statistics based algorithm~\cite{martingales} and an extended stop-and-stare algorithm~\cite{mssa} haven been proposed to scale up the influence maximization problem with approximation guarantees based on propagation models. However, in order to implement any of them, edge weights must be pre-assigned. To eliminate such a need, the Credit Distribution (CD) model~\cite{12} was proposed to measure the influence merely based on the history of user behaviors. Following~\cite{12}, some extended CD models have been proposed to improve the estimation accuracy over the total number of people that finally get influenced by introducing node features~\cite{13} and time constraints~\cite{14}, respectively. 

The aforementioned CD based models can be trained by datasets or event logs composed of user indices, actions, and timestamps. However, the datasets used for the CD based models in existing works~\cite{12,13,14} usually have a simplified structure such that they only record one timestamp of a certain action for each user. By using such datasets, they also implicitly assume that each user takes the same action for at most once. It is easy to see that such a setup is oversimplified, since a user may take the same action multiple times. Moreover, the user who repeatedly performs a certain action is potentially more influential than the one who just performs the same action once. 

This issue can be easily verified in social networks like Twitter~\cite{15} or Facebook, where users may participate in the discussion on some topics more than once. 
\begin{figure}
   \centering
   \includegraphics[width=1\linewidth]{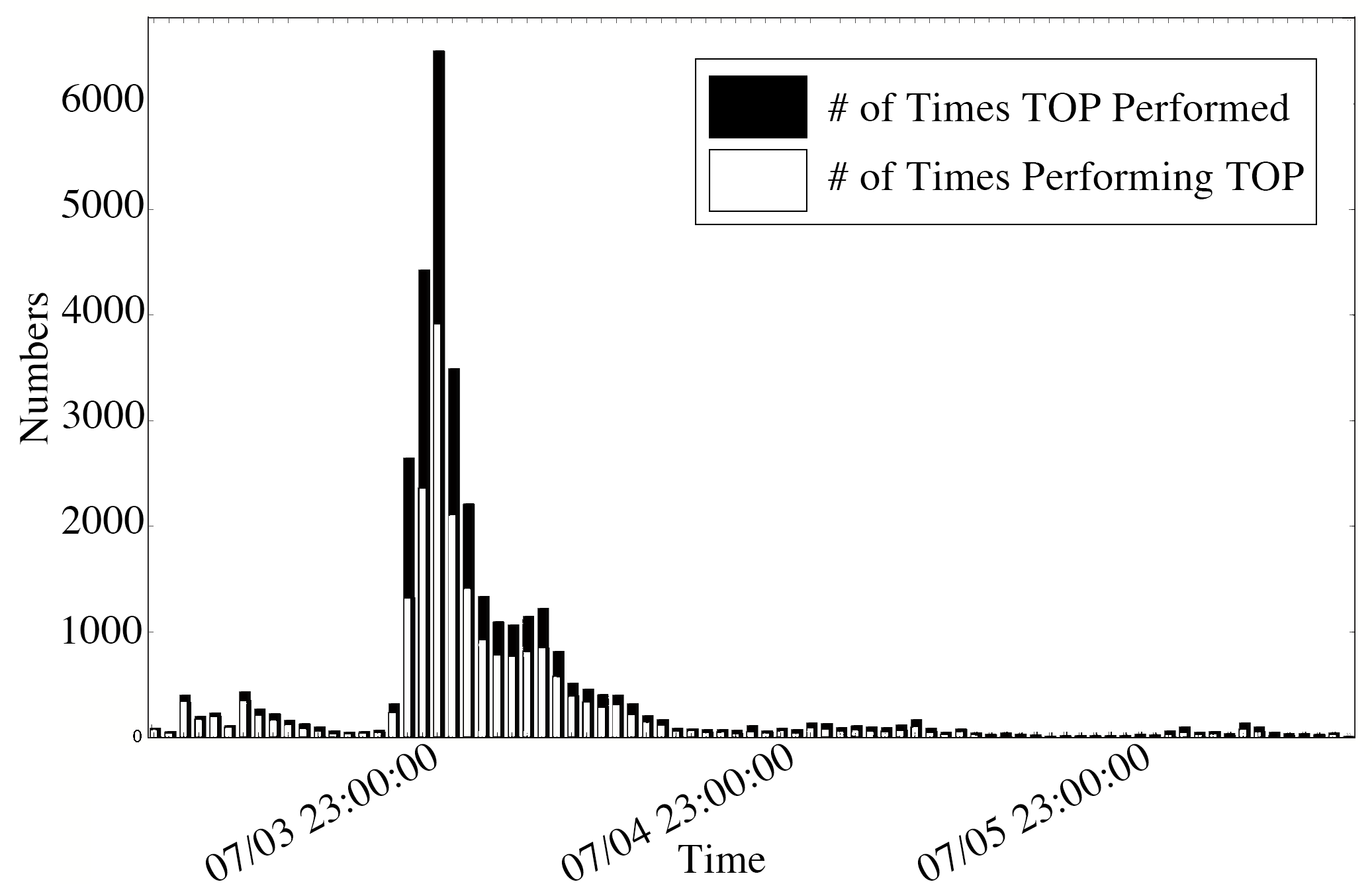} 
   \caption{Difference between Action Frequency and Numbers of Users Involved.}
   \label{fig:top1}
\end{figure}
In Fig. \ref{fig:top1}, we investigate an action named ``TOP", where we compute over time how frequently the ``TOP" action is taken and the number of users taking this action within the given time interval. It can be observed that they both follow the same decay fashion and there is a big difference between the two sequences of bars. For example, from 23:00:00 to 23:59:59 on July 3rd, there are about 4,000 people taking the ``TOP" action, while such an action is performed more than 6,500 times. The observed difference implies that some users indeed take the action ``TOP'' for multiple times. To further quantify how the users repeat the same action, we define the repetition rate of action $a$ as 
\begin{align*} 
  1 - \frac{\text{Number of Users Performing Action }a}{\text{Number of Times Action } a \text{ is Performed}}.
\end{align*}

The value of repetition rate refers to the percentage of the executions that are not performed for the first time over the total number of executions. For the 100 most common actions in the dataset, we find that there are about 43\% of the actions with repetition rates over 10\% as shown in Fig. \ref{fig:rpr}. Therefore, it is useful to develop a model that can take the repetition of actions into account.
\begin{figure}
   \centering
   \includegraphics[width=1.8in]{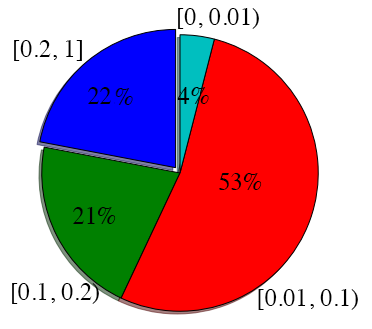} 
   \caption{Repetition Rate of the 100 Most Common Actions.} 
   \label{fig:rpr}
\end{figure}

In this paper, we perform data analysis on a multi-action event log, where the same action for one particular user is recorded for multiple times if the user performs this action repeatedly. To deal with such an event log, we propose a novel \emph{multi-action credit distribution} (mCD) model, which uses the time-dependent ``credit'' to quantify the influence ability for each user. Based on this model, we formulate a budgeted influence maximization problem, which aims to identify a subset of users with the maximum influence ability of the selected subset. In this problem, the objective function, i.e., the influence ability, is submodular; and a knapsack constraint is added to regulate the cost for user selection. This problem is NP-hard. 

To solve this problem, we first consider a simplified case with a cardinality constraint. By utilizing submodularity, we develop an efficient streaming algorithm that scans through the user set only once and uses a certain threshold to decide whether a user should be added to the solution set. Such a streaming algorithm is within ($\frac{1}{2}-\epsilon$) of the optimality. Then, we modify the algorithm to solve the general knapsack case, which can guarantee ($\frac{1}{3}-\epsilon$) of the optimality. Experimental results over real Twitter dataset show: 1) Compared to the existing CD and non-CD models, the mCD model is more accurate in estimating the total number of people that finally get influenced by the selected set; and 2) under the mCD model, the proposed streaming algorithms can achieve a similar utility as the greedy algorithm CELF~\cite{11}, while having a much faster computation speed.

The rest of this paper is organized as follows. In Section~2, we describe the mCD model along with the formulation of the influence maximization problem. In Section~3, we introduce a learning algorithm to train the mCD model with the given event log, and present two streaming algorithms to solve the influence maximization problem under a cardinality  constraint and a knapsack constraint, respectively. In Section~4, we numerically demonstrate the performance of the proposed mCD model and the corresponding streaming algorithms over real data collected from Twitter. Finally, we conclude the paper in Section~5.

\section{System Model and Problem Formulation}
In this paper, an online social network is modeled as an unweighted directed graph $\mathcal{G} = (\mathcal{V}, \mathcal{E})$~\cite{12,13}, where the node set $\mathcal{V}$ is the set of all users and the edge set $\mathcal{E}$ indicates the social relationship among all the users. Specifically, for any $u,v \in \mathcal{V}$, there is a directed edge $(v,u)$ (from $v$ to $u$) if $v$ is socially followed by $u$, which implies that $v$ could potentially ``influence'' $u$. The collected data from this social network is a multi-action event log $\mathcal{L}$ with records in the form of (USER, ACTION, TIME), where a corresponding tuple ($u, a, t$)~$\in \mathcal{L}$ indicates that user $u$ performed action $a$ at time $t$. The action $a$ is from a finite action set~$\mathcal{A}$. Here, each action $a$ corresponds to a specific discussion topic. A user $u$ performed action $a$ means that he/she got involved in the discussion of that topic.

\subsection{Multi-Action Credit Distribution} 
In conventional CD model~\cite{12}, the main idea is to assign ``credits'' to the possible influencers according to the event log. The total assigned credits to a user consist of direct credits and indirect credits. If the neighbors of user $u$ perform certain action by following $u$, direct credits are then assigned to user $u$ by its neighbors. If the users get influenced by $u$ through multiple hops, indirect credits will be assigned to $u$. The value of indirect credits are computed along all possible trails. The conventional CD model can effectively quantify the influence ability of each user in a single-action event log, but not for the case with a multi-action event log. Next, we introduce the detailed design of the mCD model that can handle multi-action event logs.

Suppose that in a multi-action event log $\mathcal{L}$, it is recorded that user $u$ performs action $a$ for $A_u(a)$ times. For user-action pair $(u,a)$, if $A_u(a) \geq 1$, let $t_i(u,a)$ denote the timestamp when user $u$ performs acton $a$ for the $i$-th time; otherwise, the timestamp is not needed. Next, we denote $\mathcal{A}_u$ as the set of actions that are performed by user $u$. Note that the conventional CD model is a special case of the proposed mCD model, in which $A_u(a) \in \left\{ 0,1 \right\}$ for all $u \in \mathcal{V}$ and $a \in \mathcal{A}$. Based on the directed graph $\mathcal{G}$ and the multi-action event log $\mathcal{L}$, for any action $a \in \mathcal{A}$, we define a directed graph $\mathcal{G}(a)$ that is generated from  $\mathcal{G}$ according to the propagation of action~$a$. Specifically, we define $\mathcal{G}(a) = (\mathcal{V}(a), \mathcal{E}(a))$ such that $\mathcal{V}(a) = \{ v \in \mathcal{V} | A_v(a) \geq 1\}$  and $\mathcal{E}(a) = \{(v,u) \in \mathcal{E} | t_1(v,a) < t_1(u,a), A_u(a) \cdot A_v(a) \geq 1\}$. Then, for any user $u$ who performs action $a$, we let $\mathcal{N}_{in}(u,a) = \{v | (v,u) \in \mathcal{E}(a)\}$ denote the set of direct influencers for user $u$, i.e., the neighbors of user $u$ who perform action $a$ earlier than user $u$. Furthermore, we denote by $\mathcal{N}_{in}(\mathcal{S},a) = \{v | v \in \mathcal{N}_{in}(u,a), u \in \mathcal{S}, v \notin \mathcal{S} \}$ as the neighborhood of a given user set $\mathcal{S}$ with respect to action~$a$.

For a given action $a$, we define a timestamp set $\mathcal{T}_{v,u}(a)=\{t_i(v,a)|t_i(v,a)<t_1(u,a), 1\leq i \leq A_v(a)\}$ for every pair of users $u$ and $v$ such that $u \in \mathcal{V}(a)$ and $v\in \mathcal{N}_{in}(u,a)$, which is a collection of timestamps of $v$ performing action $a$ before user $u$. Intuitively, each time when user $v$ performs the action, it causes influence on user $u$, since $v$ and $u$ have a directed edge $(v,u)$ in $\mathcal{G}(a)$. To take this effect into consideration, we consider a series of delays that can be expressed by the timestamp differences, i.e., $t_1(u,a)-t$, for all $t \in \mathcal{T}_{v,u}(a)$. Note that the conventional CD model just simply uses one delay to calculate the direct credit. Here, on the other hand, we adopt an the effective delay from $v$ to $u$ on action $a$, which is defined as
\begin{align} 
\label{eq:h_mean}
   \Delta t_{v,u}(a) = \frac{1}{\sum_{t \in \mathcal{T}_{v,u}(a)}{(t_1(u,a) - t)^{-1}}}.
\end{align}
Note that $\Delta t_{v,u}(a)$ equals the harmonic mean of $\{(t_1(u,a)-t)\}$ devided by $|\mathcal{T}_{v,u}(a)|$. There are some useful properties of $\Delta t_{v,u}(a)$: 1) $\Delta t_{v,u}(a)\le \min\{(t_1(u,a)-t)\}$ for $t \in \mathcal{T}_{v,u}(a)$, and 2) $\Delta t_{v,u}(a)$ decreases as  $|\mathcal{T}_{v,u}(a)|$ increases.

The definition of $\Delta t_{v,u}(a)$ is inspired by the calculation of parallel resistance, where the effective resistance of multiple parallel resistors is mainly determined by the smallest one, and whenever a new resistor is added in parallel, the effective resistance decreases. Similarly, while every time user $v$ taking action $a$ poses some influence on user $u$, it is reasonable to assume that the most recent action induces the most significant influence. Thus, it is desired that the value of the effective delay $\Delta t_{v,u}(a)$ is dominated by $\min \{ t_1(u,a)-t | t\in \mathcal{T}_{v,u}(a) \}$. In addition, a user $u$ would be more likely to follow his neighbor $v$ on action $a$ if $v$ takes action $v$ for many times. In other words, by repeatedly taking the same action, user $v$ poses stronger influence on user $u$. Based on the definition of effective delay, we next define direct credit and indirect credit.

\textbf{Direct Credit.} This credit is what user $u$ assigns to its neighbor $v$ when $u$ takes the same action $a$ after $v$. The direct credit $\gamma_{v,u}(a)$ is defined as
\begin{align}
\label{eq:r}
   \gamma_{v,u}(a) =\exp{\left(-\frac{\Delta t_{v,u}(a)}{\tau_{v,u}}\right)} \cdot \frac{1}{R_{u,a}},
\end{align}
where $\tau_{v,u}$ and $R_{u,a}$ are normalization factors. Note that the direct credit decays exponentially over the effective delay $\Delta t_{v,u}(a)$. Such an exponential expression follows from the original definition of the CD model~\cite{12}. Here, $\tau_{v,u}$ is the mathematical average of the time delays between $v$ and $u$ over all the actions:
\begin{align} 
\label{eq:tau}
   \tau_{v,u} = \frac{1}{\left|\mathcal{A}_{v2u}\right|}\cdot\sum_{a \in \mathcal{A}_{v2u}}{\frac{\sum_{t \in \mathcal{T}_{v,u}(a)}{\left(t_1(u,a) - t\right)}}{|\mathcal{T}_{v,u}(a)|}}, 
\end{align}
where $\mathcal{A}_{v2u}$ denotes the set of actions that $v$ takes prior to $u$. In addition, $R_{u,a}$ is given by $$R_{u,a} = \sum_{v \in \mathcal{N}_{in}(u,a)}{\exp{\left(-{\Delta t_{v,u}(a)}/{\tau_{v,u}}\right)}},$$ 
which ensures that the direct credit assigned to all the neighbors of $u$ for action $a$ sums up to $1$.

To summarize, for $u, v \in \mathcal{V}$, the direct credit given to $v$ by $u$ with respect to action $a$ is given as
\begin{align}
\label{eq: gamma}
\gamma_{v,u}(a)=\left\{
\begin{array}{lcl}
\exp{\left(-\frac{\Delta t_{v,u}(a)}{\tau_{v,u}}\right)} \cdot R_{u,a}^{-1},      &      & {(v,u) \in \mathcal{E}(a)}\text{;} \\
0,       &      & {\text{otherwise.}}
\end{array} \right. 
\end{align}

\textbf{Indirect Credit.} 
Suppose that $(v,w)$ and $(w,u)$ are in $\mathcal{E}(a)$ such that $v$ and $u$ are connected indirectedly. Then, user $u$ assigns indirect credit to $v$ via $w$ as $\gamma_{v,w}(a) \cdot \gamma_{w,u}(a)$. As such, the total credits given to $v$ by $u$ on action $a$ can be defined iteratively as
\begin{align}
\label{eq:Ga} 
   \Gamma_{v,u}(a) = \sum_{w \in \mathcal{N}_{in}(u,a)}{\Gamma_{v,w}(a)\cdot \gamma_{w,u}(a)},
\end{align}
where $\Gamma_{v,v}(a) = 1$. Then, the average credit given to $v$ by $u$ with respect to all actions is defined as:
$$ \kappa_{v,u} =\left\{
\begin{array}{lcl}
0,       &      &|\mathcal{A}_u|=0;\\
\frac{1}{|\mathcal{A}_u|}\sum_{a \in \mathcal{\mathcal{A}}}{ \Gamma_{v,u}(a)},      &      & {\text{otherwise}}.
\end{array} \right. $$

Moreover, for a set of influencers $\mathcal{S}\subseteq \mathcal{V}(a)$ on action $a$, we have
$$ \Gamma_{\mathcal{S},u}(a)=\left\{
\begin{array}{lcl}
1,       &      & {u \in \mathcal{S}};\\
\sum_{w \in \mathcal{N}_{in}(u,a)}{\Gamma_{\mathcal{S},w}(a)\cdot \gamma_{w,u}(a)},      &      & {\text{otherwise}}.
\end{array} \right. $$
Similarly, we define the average credit given to $\mathcal{S}$ by $u$ with respect to all the actions as: 
$$  \kappa_{\mathcal{S},u}  =\left\{
\begin{array}{lcl}
0,       &      &|\mathcal{A}_u|=0;\\
\frac{1}{|\mathcal{A}_u|}\sum_{a \in \mathcal{\mathcal{A}}}{ \Gamma_{\mathcal{S},u}(a)},      &      & {\text{otherwise}}.
\end{array} \right. $$

Note that the average credit $\kappa_{\mathcal{S},u}$ can also be interpreted as the ``influence ability'' of the set $\mathcal{S}$ on a particular user $u$, and the value of $\kappa_{\mathcal{S},u}$ indicates how influential $\mathcal{S}$ is. Finally, we define $\sigma_{mCD}(\mathcal{S})$ as the influence ability of $\mathcal{S}$ over the whole network, which is given as
\begin{align} 
\label{objective}
   \sigma_{mCD}(\mathcal{S}) = \sum_{u \in \mathcal{V}}{\kappa_{\mathcal{S},u}}.
\end{align}
\textbf{Remark:} $\sigma_{mCD}(\mathcal{S})$ is monotone and submofular. To see this point, it is sufficient to show that $\Gamma_{\mathcal{S}, u}(a)$ is monotone and submodular for any $u\in \mathcal{V}$ and $a\in \mathcal{A}$, since a positive linear combination of monotone and submodular functions is still monotone and submodular~\cite{17}. As the propagation graph $\mathcal{G}(a)$ shares the same acyclic property, similar to the argument in~\cite{12}, we can show that $\Gamma_{\mathcal{S}, u}(a)$ is monotone and submodular by induction. First, we restrict the attention path of $\Gamma_{\mathcal{S}, u}(a)$ with length $0$ by the definition of submodularity, where the attention path is introduced to limit the indirect credit calculation through some path with a length less than a given value. Then, by assuming that submodularity holds when the attention path equals $l$, we can easily show the submodularity for the case of $l+1$. Since the maximum length of the attention path is $|\mathcal{V}|-1$, we could reach the conclusion that $\Gamma_{\mathcal{S}, u}(a)$ is monotone and submodular. Therefore, we could claim that the influence ability function $\sigma_{mCD}(\mathcal{S})$ is monotone and submodular.

\subsection{Budgeted Influence Maximization Problem}
A budgeted influence maximization problem in a social network can be formulated as finding a subset $\mathcal{S}$ of users, i.e., a seed set, from the ground set $\mathcal{V}$ to achieve maximum influence ability within some user selection budget. Note that for a general case, different user selection criteria may lead to different cost. For example, if the users are chosen and paid to spread out certain advertising information, one would expect that a user with more fans charges more than others. This is reasonable since the value of a user is related to how many people he or she can potentially influence over the network. Therefore, we introduce a knapsack constraint to quantify the user selection cost. Suppose there are $n$ users in the dataset. We denote a positive $n \times 1$ weight vector $g = (g_1, g_2, \ldots, g_n)^T$ as the unique cost for selecting each user. Denote $I_\mathcal{S} = (I_{1}, I_{2}, \cdots, I_{n})^T$ as an $n \times 1$ characteristic vector of $\mathcal{S}$, where $I_i = 1$ if $i \in \mathcal{S}$; $I_i = 0$, otherwise. Let $b$ be the total available budget on the cost for selecting users into $\mathcal{S}$. Then, the budgeted influence maximization problem could be cast as

\begin{eqnarray}\label{overallproblem}
\begin{aligned}
   &\underset{\mathcal{S} \subseteq \mathcal{V}}{\textrm{maximize}} &&\sigma_{mCD}(\mathcal{S}) \\ 
   &\textrm{subject to} && g^TI_{\mathcal{S}} \leq b.
\end{aligned}
\end{eqnarray}
For simplicity, we normalize problem~(\ref{overallproblem}) as follows. We first divide the knapsack constraint by the minimum weight $g_{\min}=\min \left\{g_i\right\}_{i=1}^n$ on both sides, i.e., $g^TI_\mathcal{S} / g_{\min}\leq b/g_{\min}$. We then treat $g/ g_{\min}$ and $b/g_{\min}$ as a new weight vector $g$ and a new budget constraint $b$ correspondingly, with a slight misuse of notations. After this manipulation, every entry in $g$ is no less than $1$ and the number of selected users will not exceed $b$. It is easy to see that the standardized problem has the same optimal solution as the original problem~(\ref{overallproblem}). For the rest of this paper, we only consider the standardized problem.

With the formulation of the budgeted influence maximization problem, it is worth noting that $\sigma_{mCD}(\mathcal{S})$ is a lower bound of the total number of users that finally get influenced over all actions, as given in Proposition~\ref{prop:1}. 
\begin{proposition} \label{prop:1}
$\sigma_{mCD}(\mathcal{S}) \leq |\cup_{a \in \mathcal{A}}\mathcal{V}(a)|.$ 
\end{proposition}

\begin{proof}
Given an action $a\in \mathcal{A}$ and a user $u \in \mathcal{V}(a)$, let $h$ denote the maximum hops that $u$ distributes credits to $\mathcal{N}_{in}(\mathcal{S},a)$. Then, the total credit can be expressed as
\begin{align*} 
\begin{split}
   \Gamma_{\mathcal{S},u}(a) = \sum_{w_1 \in \mathcal{N}_{in}(u,a)}\gamma_{w_1, u}(a) \left( \sum_{w_2 \in \mathcal{N}_{in}(w_1,a)}\gamma_{w_2, w_1}(a) \cdot \right. \\
  \left. \left(  \cdots \left( \sum_{v \in \mathcal{N}_{in}(w_h, a)} \gamma_{v,w_h}(a) \cdot \Gamma_{\mathcal{S},v}(a) \right)\right) \right). 
\end{split}
\end{align*}

According to the definition of direct credits, for any $v \in \mathcal{V}$ and $a \in \mathcal{A}$, we have a normalizer $R_{v,a}$ to ensure $$\sum_{v' \in \mathcal{N}_{in}(v,a)} \gamma_{v',v}=1.$$ Thus, for $v \in \mathcal{N}_{in}(w_h, a)$, we have 
\begin{align*} 
\Gamma_{\mathcal{S}, v}(a) &= \sum_{v' \in \mathcal{S} \cap \mathcal{N}_{in}(v,a)} \gamma_{v',v}(a)\\ &\leq \sum_{v' \in \mathcal{N}_{in}(v,a)} \gamma_{v',v}(a)=1, 
\end{align*}
and then,
$$\sum_{v \in \mathcal{N}_{in}(w_h, a)} \gamma_{v,w_h} \cdot \Gamma_{\mathcal{S},v}(a) \leq 1.$$

Analogously, we can show that the total credit given by any $u\in \mathcal{V}(a)$ on action $a \in \mathcal{A}$ is bounded, i.e., $\Gamma_{\mathcal{S},u}(a) \leq 1$. By the definition of influence ability, we have
\begin{align*} 
    \sigma_{mCD}(\mathcal{S}) &= \sum_{u \in \cup_{a \in \mathcal{A}}\mathcal{V}(a)} \frac{1}{|\mathcal{A}_u| } \sum_{a \in \mathcal{A}} \Gamma_{\mathcal{S},u}(a) \\&=  \sum_{u \in \cup_{a \in \mathcal{A}}\mathcal{V}(a)} \frac{1}{|\mathcal{A}_u| } \sum_{a \in \{a| a \in \mathcal{A}, u_u(a) \geq 1\}} \Gamma_{\mathcal{S},u}(a) \\ &= \sum_{u \in \cup_{a \in \mathcal{A}}\mathcal{V}(a)} \frac{1}{|\mathcal{A}_u| } \sum_{a \in \mathcal{A}_u} \Gamma_{\mathcal{S},u}(a) \\ &\leq \sum_{u \in \cup_{a \in \mathcal{A}}\mathcal{V}(a)} \frac{1}{|\mathcal{A}_u| } \cdot |\mathcal{A}_u| \\ &= |\cup_{a \in \mathcal{A}}\mathcal{V}(a)|.
\end{align*}

Note that when we evaluate a single action $a$, $\sigma_{mCD} \left(\mathcal{S} \right)$ provides a 
lower bound of $\mathcal{V}(a)$.

\end{proof}

Therefore, problem~(\ref{overallproblem}) is to find a subset $\mathcal{S}$ from the ground set $\mathcal{V}$ to maximize a lower bound of  the total number of users that finally get influenced over all actions. In Section~4, we will numerically show that the influence ability $\sigma_{mCD}(\mathcal{S})$ in the mCD model, although as a lower bound, provides a more accurate approximation of $|\mathcal{V}(a)|$ with respect to each action $a \in \mathcal{A}$, compared to the influence ability in the CD model. 

As aforementioned, the objective function of problem (\ref{overallproblem}) is monotone and submodular. Therefore, problem~(\ref{overallproblem}) is a submodular maximization problem under a knapsack constraint, which has been proved to be NP-hard~\cite{16}. In general, such a submodular problem can be approximately solved by greedy algorithms~\cite{10,16}. However, due to the large volume of online social network datasets, the implementation of greedy algorithms is not practical. In the next section, we develop an efficient streaming algorithm to solve the budgeted influence maximization problem under the mCD model.

\section{Algorithms}
The proposed algorithm is divided into the following modules, as shown in Fig.~\ref{fig:system_model}. The module ``Model Learner'' is designed to learn the parameters $\{\tau_{v,u}\}$ (the mathematical average time delay between each pair of $v$ and $u$ over all actions) and $\{A_u(a)\}$ (the frequency of $u$ taking action $a$), from the training dataset before solving the optimization problem, such that the algorithm can deal with a newly arriving dataset or test set much more efficiently. Then, for the new or test set of data, we start with the preprocessing module ``Log Scanner'', which scans the dataset to calculate the total credit $\Gamma_{v,u}(a)$ assigned to user $v$ by $u$ for action $a$ by using the already learned $\{\tau_{v,u}\}$ and $\{A_u(a)\}$ from the training set. The last but the most important module  ``Problem Solver'' solves the influence maximization problem~(\ref{overallproblem}) based on $\{\Gamma_{v,u}(a)\}$ and outputs the seed set. 
\begin{figure}[H]
   \centering
   \includegraphics[width=3in]{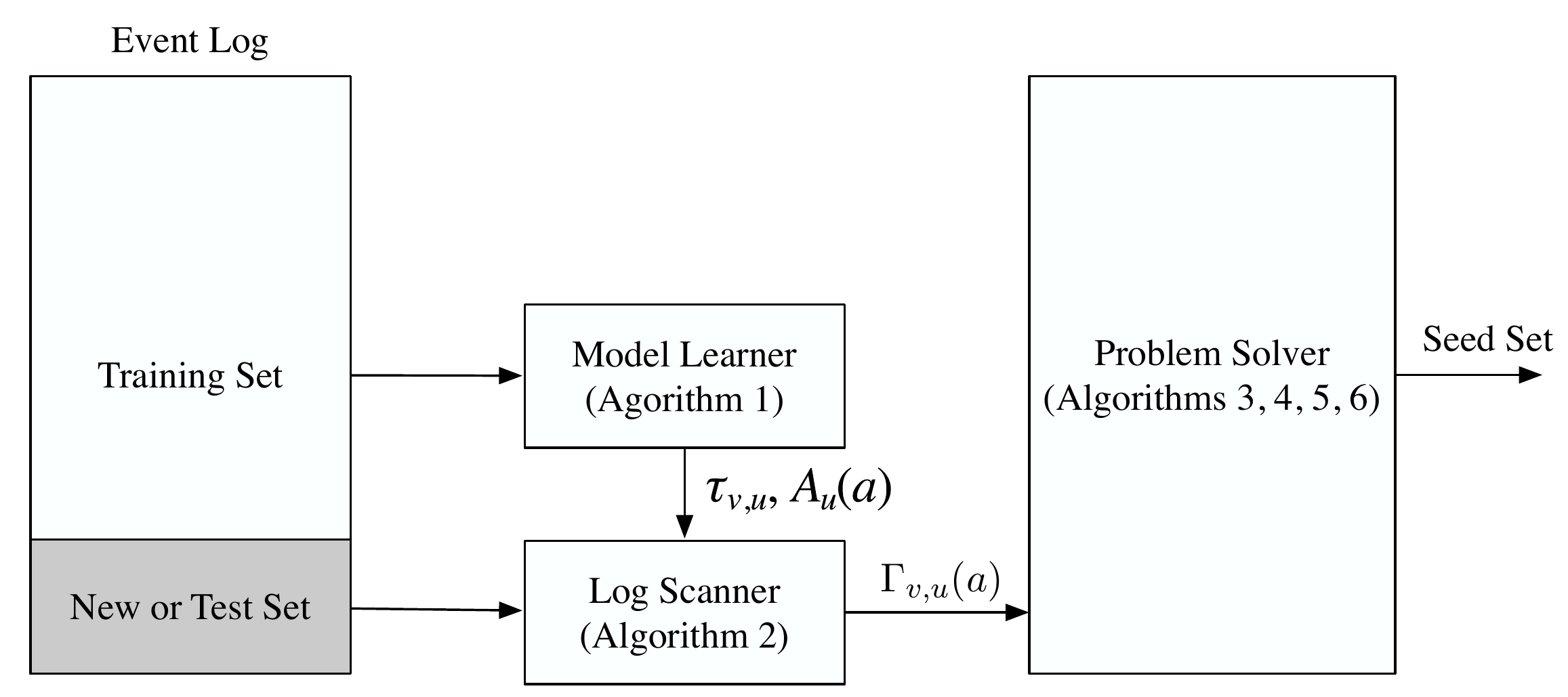} 
   \caption{Overall Algorithm Structure for Influence Maximization.}
   \label{fig:system_model}
\end{figure}

\subsection{Parameter Learning}
The main function of Model Learner is to learn $\{\tau_{v,u}\}$ and $\{A_u(a)\}$ from the event log, where $\tau_{v,u}$ is mainly determined by $\mathcal{T}_{v,u}(a)$ and $A_{v2u} =|\mathcal{A}_{v2u}|$, according to Eq.~(\ref{eq:tau}). Since $\mathcal{T}_{v,u}(a)$ can be directly constructed from the event log according to its definition in Section II-A, the key problem is to compute $A_{v2u}$, or equivalently to find all parents of $u$ for a particular action $a$.  Here, we propose Algorithm~\ref{Learning} to obtain $\{\tau_{v,u}\}$ by computing  $\{A_{v2u}\}$.

\begin{algorithm}
\caption{MODEL LEARNER}
\label{Learning}
\begin{algorithmic}[1]
\State \textbf{Initialize} $A_{v2u} := 0$ for all users and edges.
\NoDoFor{each action $a$ in training set}
	\State $current\_table := \emptyset$.
	\State $A_u(a) := 0$ for all users.
	\NoDoFor{each tuple $<u,a,t>$ in chronological order}
		\State \textbf{if} $A_u(a) \neq 0$, \textbf{then} continue. 
		\State $parents(u,a) := \emptyset; A_u(a) := A_u(a) + 1$.	
		\NoDoWhile{$\exists v : (v,u) \in \mathcal{E}$ and $v \in current\_table$}
			\State $parents(u,a) := parents(u,a) \cup \{v\}$.
		\EndWhile
		\State $A_{v2u} := A_{v2u} + 1$,  $\forall v \in parents(u,a)$.
		\State $current\_table := current\_table \cup \{u\}$.
	\EndFor
\EndFor
\State \textbf{Output} $\tau_{v,u}$ according to Eq. (\ref{eq:tau}), $A_u(a)$.
\end{algorithmic}
\end{algorithm}

Specifically, $current\_table$ is maintained to store user indices who have performed $a$ and have been scanned so far; and $parents(u,a)$ is a list of parents of $u$ with respect to the action $a$. The incremental update process for $A_{v2u}$ is repeated with respect to each action $a$, in order to compute the total number of actions propagated from $v$ to $u$. At the end of Algorithm~\ref{Learning}, we use the definition in Eq. (\ref{eq:tau}) to compute the average time delay between every valid user pair.

\subsection{Problem Solving}
\subsubsection{Computation of Marginal Gain}
As problem (\ref{overallproblem}) is NP-hard, it is not practical to obtain the optimal solution over large datasets. For such a problem, a greedy algorithm was proposed in~\cite{17} to obtain a suboptimal solution with a factor ($1-1/e$) away from optimality. At each step, the greedy algorithm scans over all the unselected users, and picks the user with the largest marginal gain. The drawback of the greedy algorithm is that scanning over all the unselected users repeatedly is very time-consuming, especially when the dataset is large. In this section, we propose an alternative way to calculate the marginal gain efficiently.

Based on $\{\tau_{v,u}\}$ and $\{A_u(a)\}$ obtained by ``Model Learner'', we can compute the marginal gain based on the definition of the influence ability given in Eq. (\ref{objective}) directly. However, it requires the computation of the total credit $\{\Gamma_{v,u}(a) \}$ for each user as well as the total credit for each pair of neighbors, which is quite inefficient under a big data setting. Thus, we adopt the following alternative and equivalent method to calculate the marginal gain. First, denote by $\Gamma_{x,u}^{\mathcal{V} - \mathcal{S}}(a)$ the total credit given to $x$ by $u$ on action $a$ through the paths that are contained completely in the subgraph induced by $\mathcal{V} - \mathcal{S} = \left\{v \in \mathcal{V}| v \notin \mathcal{S} \right\}$. Note that when $\mathcal{S}$ is the null set, we have $\Gamma_{x,u}^{\mathcal{V} - \mathcal{S}}(a) = \Gamma_{x,u}(a)$ as defined in Eq. (\ref{eq:Ga}). For the subgraphs, the following lemmas hold.
\begin{lemma}
\label{gamma_v_u}
$\Gamma_{v,u}^{\mathcal{S}-x}(a) = \Gamma_{v,u}^{\mathcal{S}}(a) - \Gamma_{v,x}^{\mathcal{S}}(a) \cdot \Gamma_{x,u}^{\mathcal{S}}(a)$.
\end{lemma}
\begin{lemma}
\label{gamma_S_u}
$\Gamma_{\mathcal{S}+x,u}(a) = \Gamma_{\mathcal{S},u}(a) + \Gamma_{x,u}^{\mathcal{V}-\mathcal{S}}(a) \cdot (1-\Gamma_{S,x}(a))$. 
\end{lemma}
Then, we have the following theorem to compute the marginal gain.
\begin{theorem}
\label{th:mg}
In the mCD model, given any subset $\mathcal{S} \subseteq \mathcal{V}$ and an element $x \in \mathcal{V}-\mathcal{S}$, the marginal gain of adding $x$ into $\mathcal{S}$ equals
	$$\sum_{a\in A}\left((1-\Gamma_{\mathcal{S},x}(a)) \cdot \sum_{u \in \mathcal{V}}{\frac{1}{A_u} \cdot \Gamma_{x,u}^{\mathcal{V}-\mathcal{S}}(a)}  \right).$$
\end{theorem}
The proof of the above lemmas and theorem can be easily obtained by results in~\cite{12}, which is omitted here.~With these observations, when we add a new user $x$ into $\mathcal{S}$, we see that we do not need~to iteratively calculate  $\sigma_{mCD}(\mathcal{S}+x)$ and $\sigma_{mCD}(\mathcal{S})$. Instead, we keep updating $\Gamma_{v,u}^{\mathcal{V-S}}(a)$ and~$\Gamma_{\mathcal{S},u}(a)$ using Lemmas~\ref{gamma_v_u} and~\ref{gamma_S_u}, after which we can compute the marginal gain with Theorem~\ref{th:mg}. In Algorithm~\ref{UC} below, ``Log Scanner'' scans over the test set to calculate $\Gamma_{v,u}(a)$  for every user pair $(v,u)$ on every action $a$, and stores the result in $UC[a][v][u]$. This module provides the initialization for the later-on "Problem Solver" module. 
\begin{breakablealgorithm}
\caption{LOG SCANNER}
\label{UC}
\begin{algorithmic}[1]
\State \textbf{Initialize }$ UC[a][v][u] := 0$ for all actions and users.
\NoDoFor{each action $a$ in $\mathcal{A}$}
	\State $current\_table := \emptyset$.
	\State $A_u(a):=0$ for all users. 
	\NoDoFor{each tuple $<u,a,t>$ in chronological order}
		\State \textbf{if} $A_u(a) \neq 0$, \textbf{then} continue.
		\State $parents(u) := \emptyset; A_u(a) := A_u(a) + 1$.	
		\NoDoWhile{$\exists v : (v,u) \in \mathcal{E}$ and $v \in current\_table$}
			\State $parrents(u) := parrents(u) \cup \{v\}$.
		\EndWhile
		\NoDoFor{each $v \in parrents(u)$}
			\State Compute $\gamma_{v,u}(a)$ according to Eq. (\ref{eq: gamma}).
			\State $UC[a][v][u] := UC[a][v][u] + \gamma_{v,u}(a)$.
			\State $UC[a][w][u] := UC[a][w][u] + UC[a][w][v] \cdot \gamma_{v,u}(a),$ $\forall w \in \mathcal{V}$.			
		\EndFor
		\State $current\_table := current\_table \cup \{u\}$.
	\State $UC[a][v][v] := 1$, $\forall v \in current\_table$.
	\EndFor
\EndFor
\State \textbf{Output} $UC$.
\end{algorithmic}
\end{breakablealgorithm}

With the output $UC$ by Algorithm \ref{UC}, we are ready to compute the marginal gain in the ``Problem Solver" module. The structure of the ``Problem Solver" module can be summarized as follows. In each iteration over the users, if the marginal gain of a candidate user satisfies a particular criterion\footnote{The details of the condition to select candidate users will be explained in Algorithms~\ref{str_k} and \ref{str_b}.}, this user will be added to the seed set. The core of Problem Solver is Algorithm \ref{mg}, which relies on Theorem \ref{th:mg} to compute the marginal gain. In particular, Algorithm~\ref{mg} takes the candidate user $x$, user credit $UC$, and set credit $SC$ as inputs, and returns the marginal gain $mg$ for adding node $x$ into the seed set. Here, the structure of $SC$ is indexed by $[a][x]$; it stores the total credit $\{\Gamma_{\mathcal{S}, x}(a)\}$ given to the current seed set $\mathcal{S}$ by a user $x$ for an action $a$, and is updated whenever a new user is added to the seed set $\mathcal{S}$. Once user $x$ is added to the seed set, $\Gamma_{\mathcal{S},x}(a)$ and $\Gamma_{x,u}^{\mathcal{V}-\mathcal{S}}(a)$ are updated according to Algorithm \ref{update}.

\begin{breakablealgorithm}
\caption{COMPUTE\_MARGINAL\_GAIN($x,UC,SC$)}
\label{mg}
\begin{algorithmic}[1]
\State \textbf{Initialize} $mg := 0; mg_a := 0$ for all actions.
\NoDoFor{each action $a$ such that $\exists u : UC[a][x][u] > 0$}
	\NoDoFor{each $u$ such that $UC[a][x][u] > 0$}
		\State $mg_a := mg_a + UC[a][x][u] / u_u(a)$.	
	\EndFor
	\State $mg := mg + mg_a \cdot (1- SC[a][x])$.
\EndFor
\State \textbf{return} $mg$.
\end{algorithmic}
\end{breakablealgorithm}
\begin{breakablealgorithm}
\caption{UPDATE($x, UC, SC$)}
\label{update}
\begin{algorithmic}[1]
\State $UC_{old} = UC$, $SC_{old} = SC$
\NoDoFor{each action $a$}
	\NoDoFor{each $u$}
		\State $UC[a][v][u] := UC_{old}[a][v][u] - UC_{old}[a][v][x] \cdot UC_{old}[a][x][u], \forall v \in \mathcal{V}$.
	\State $SC[a][u] := SC_{old}[a][u] + UC_{old}[a][x][u] \cdot (1- SC_{old}[a][x])$.	
	\EndFor
\EndFor
\State \textbf{return} $UC, SC$.
\end{algorithmic}
\end{breakablealgorithm}


\subsubsection{Influence Maximization Problem Solver}
With the algorithms to efficiently compute the marginal gain and update the total credits, we now arrive at the design of the streaming algorithms to solve problem~(\ref{overallproblem}). 

We start with a special case of the knapsack constraint, which is a cardinality constraint (by applying the same weight for every user). Given $k$ as the cardinality limit for $\mathcal{S}$, this simplified problem is cast as
\begin{eqnarray}\label{cardinalityproblem}
\begin{aligned}
   &\underset{\mathcal{S} \subseteq \mathcal{V}}{\textrm{maximize}} &&\sigma_{mCD}(\mathcal{S}) \\
   &\textrm{subject to} && |\mathcal{S}| \leq k.
\end{aligned}
\end{eqnarray}

In~\cite{18}, a streaming algorithm has been proposed to solve a submodular maximization problem under a cardinality constraint, whose main idea is to use a pre-defined threshold to justify whether a user is good enough to be selected. However, setting the threshold requires the priori knowledge of the optimal value of the problem. In most scenarios, this leads to the chicken and egg dilemma.

To address this issue, we adapt the threshold along the process instead of using a fixed threshold based on the priori knowledge of the optimal value. First, we assume that the maximum influence ability that can be achieved by any user $x\in\mathcal{V}$ is known as $m$ (we will remove this assumption later in this section), where $m = \max_{x \in \mathcal{V}}\sigma_{mCD}(\{ x\})$; we construct an optimum value candidate set $\mathcal{O}:=\{(1+\epsilon)^i | i \in \mathbb{Z}, m \leq (1+\epsilon)^i \leq k\cdot m\}$. Since the objective function is submodular and the cardinality constraint is $k$, it is easy to see that the optimal value lies in $[m, km]$. Moreover, the optimum value candidate set $\mathcal{O}$ has a property that there exist some values close to the true optimal value, as shown by the following lemma.

\begin{lemma} \label{SetO}
Let $\mathcal{O}:=\{(1+\epsilon)^i | i \in \mathbb{Z}, m \leq (1+\epsilon)^i \leq k\cdot m\}$
 for some $\epsilon$ with $0 < \epsilon < 1$. Then there exists a value $c \in \mathcal{O}$ such that $(1-\epsilon)\text{OPT} \leq c \leq \text{OPT}$, with \text{OPT} denoting the optimal value for problem~(\ref{cardinalityproblem}).
\end{lemma}
\begin{proof}
First, we choose $x \in \mathcal{V}$ such that $\sigma_{mCD}(\{x\})=m.$ We then have $\textrm{OPT}\ge \sigma_{mCD}(\{x\})=m.$ In addition, let $\{x_1,x_2,\ldots,x_k\}$ be a subset of $V$ such that $\sigma_{mCD}(\{x_1,x_2,\ldots,x_k\})=\textrm{OPT}$. By the submodularity of $\sigma_{mCD}$, we have

\begin{align*}
\begin{split}
\textrm{OPT}&=\sigma_{mCD}(\emptyset)+\sum_{i=1}^k [\sigma_{mCD}(\{x_1,x_2,\ldots,x_i\})- \\
&\sigma_{mCD}(\{x_1,x_2,\ldots,x_{i-1}\})]\\
&\le \sigma_{mCD}(\emptyset)+\sum_{i=1}^k [\sigma_{mCD}(\{x_i\})-\sigma_{mCD}(\emptyset)]\\
&\le \sum_{i=1}^k \sigma_{mCD}(\{x_i\})\le km.
\end{split}
\end{align*}
By setting $c = [1+\epsilon]^{\left\lfloor\log_{1+\epsilon}\textrm{OPT}\right\rfloor}$, we then obtain $$\frac{m}{1+\epsilon}\le \frac{\textrm{OPT}}{1+\epsilon}\le c\le \textrm{OPT}\le km,$$ and $$c\ge \frac{\textrm{OPT}}{1+\epsilon}\ge (1-\epsilon)\textrm{OPT}.$$
\end{proof}

Therefore, by constructing $\mathcal{O}$, we are able to obtain a good estimate $c$ on OPT. However, if we do not have the knowledge on $m$, we need one more scan over the user set to obtain $m$. In Algorithm~\ref{str_k}, we design a singe pass structure where we update $m$ during the iterations over user selection. Specifically, we modify $\mathcal{O}$ as $\mathcal{O}=\{(1+\epsilon)^i | i \in \mathbb{Z}, m \leq (1+\epsilon)^i \leq 2k\cdot m\}$, and maintain the variable $m$ that holds the current maximum marginal value of all single elements when the algorithm scans over the ground set. Whenever $m$ gets updated, the algorithm updates the set $\mathcal{O}$ accordingly. For each user in the ground set, we scan each element $c$ in set $\mathcal{O}$, and add that user into $\mathcal{S}_c$ as long as the marginal gain is larger than $\frac{c}{2k}$ and $|\mathcal{S}_c|\leq k$. The computation of marginal gain is conducted by function COMPUTE\_MARGINAL\_GAIN. Once a user $x$ is added to $\mathcal{S}_c$, we update the user credit $UC_c$ and the set credit $SC_c$ with function UPDATE respectively. The performance of the described streaming algorithm (Algorithm 5) is guaranteed by Theorem~\ref{th:k}. 

\begin{theorem} \label{th:k}
Algorithm~\ref{str_k} produces a solution $\mathcal{S}$ such that $\sigma_{mCD}(\mathcal{S}) \geq \left(\frac{1}{2}-\epsilon\right)\textrm{OPT}.$
\end{theorem}

\begin{proof}
Given $c' \in \mathcal{O}$ falling into $[(1-\epsilon) \textrm{OPT}, \textrm{OPT}]$, let us discuss the following two cases for the thread corresponding to $c'$.

\underline{Case 1:} $|\mathcal{S}_c'| = k$. For $1\le i\le k$, let $u_i$ be the element added to $\mathcal{S}_c'$ in the $i$-th iteration of the for-loop. Then we obtain
\begin{align*}
   \sigma &_{mCD}(\mathcal{S}_c')=\sigma_{mCD}(\{u_1,u_2,\ldots,u_k\}) \\&\geq \sigma_{mCD}(\{u_1,u_2,\ldots,u_k\})-\sigma_{mCD}(\emptyset)\\
   &=\sum_{i=1}^k \big[\sigma_{mCD}(\{u_1,\ldots,u_i\})-\sigma_{mCD}(\{u_1,\ldots,u_{i-1}\}) \big].
\end{align*}
By the condition in Line 8 of Algorithm~1, for $1\le i\le k$, we have
$$\sigma_{mCD}(\{u_1,u_2,\ldots,u_i\})-\sigma_{mCD}(\{u_1,u_2,\ldots,u_{i-1}\})\geq\frac{c}{2k},$$
and hence
$$\sigma_{mCD}(\mathcal{S}_c')\geq \frac{v}{2k} \cdot k \geq \frac{(1-\epsilon)}{2}\textrm{OPT}.$$

\underline{Case 2:} $|\mathcal{S}_c'| < k$. Let $\bar{\mathcal{S}_c'}= \mathcal{S}^* \backslash \mathcal{S}_c'$, where $\mathcal{S}^*$ is the optimal solution to the optimization problem. For each element $a \in \bar{\mathcal{S}_c'}$, we have
\begin{align*}
	\sigma_{mCD}(\mathcal{S}_c' \cup \{a\}) - \sigma_{mCD}(\mathcal{S}_c') <\frac{v}{2k}.
\end{align*}
Since $f$ is monotone submodular, we obtain
\begin{align*}
	\sigma_{mCD}(\mathcal{S}^*)&-\sigma_{mCD}(\mathcal{S}_c') =\sigma_{mCD}(\mathcal{S}_c' \cup \bar{\mathcal{S}_c'}) -\sigma_{mCD}(\mathcal{S}_c') \\
	&\leq \sum_{a\in \bar{\mathcal{S}}}[\sigma_{mCD}(\mathcal{S}_c' \cup \{a\}) - \sigma_{mCD}(\mathcal{S}_c')] \\
	&< \frac{v}{2k} \cdot k \leq \frac{1}{2}\sigma_{mCD}(\mathcal{S}^*),
\end{align*}
which implies that
$$\sigma_{mCD}(\mathcal{S}_c') >\frac{1}{2}\sigma_{mCD}(\mathcal{S}^*)=\frac{1}{2}\textrm{OPT}\geq \frac{(1-\epsilon) }{2}\textrm{OPT}.$$

Since $\mathcal{S} = \text{argmax}_{\mathcal{S}_c, c \in \mathcal{O} 
}\sigma_{mCD}(\mathcal{S}_c)$, there is $\sigma_{mCD}(\mathcal{S})\geq \sigma_{mCD}(\mathcal{S}_{c})$ for any $c \in \mathcal{O}$. As we have shown that $\sigma_{mCD}(\mathcal{S})\geq \frac{(1-\epsilon) }{2}\textrm{OPT}$, we obtain $$\sigma_{mCD}(\mathcal{S})\geq \sigma_{mCD}(\mathcal{S}_{c'})\geq  \left( \frac{1}{2} - \epsilon \right)\text{OPT}.$$
\end{proof}

\begin{breakablealgorithm}
\caption{STREAMING\_ALGORITHM($k,UC$)}
\label{str_k}
\begin{algorithmic}[1]
\State \textbf{Initialize:} $SC[a][u] := 0$ for all actions and users; $m:=0$.
\State $\mathcal{O}:=\{(1+\epsilon)^i | i \in \mathbb{Z}\}$.
\State $\mathcal{S}_c := \emptyset, UC_c := UC$ and $SC_c := SC$ for all $c \in \mathcal{O}$. 
\NoDoFor{each $x \in \mathcal{V}$}
	\State $m:=\max\{m, \sigma_{mCD}(\{x\})\}$
	\State $\mathcal{O}:=\{(1+\epsilon)^i | i \in \mathbb{Z}, m \leq (1+\epsilon)^i \leq 2k\cdot m\}$.
	\NoDoFor{$c \in \mathcal{O}$}
		\NoThenIf{COMPUTE\_MAGINAL\_GAIN($x,UC_c,SC_c$) $\geq \frac{c}{2k}$ and $|\mathcal{S}_c| < k$}
			\State $\mathcal{S}_c := \mathcal{S}_c \cup \{x\}$.
			\State $<UC_c, SC_c> :=$ UPDATE($x$, $UC_c$, $SC_c$)
		\EndIf
	\EndFor
\EndFor
\State \textbf{return} $\mathcal{S} := \text{argmax}_{\mathcal{S}_c, c \in \mathcal{O}
}\sigma_{mCD}(\mathcal{S}_c)$.
\end{algorithmic}
\end{breakablealgorithm} 

\subsubsection{Budgeted Influence Maximization Problem Solver}

Now, we consider the more general budgeted influence maximization problem (problem~(\ref{overallproblem})), we first modify the threshold in line~6 of Algorithm~\ref{str_k} to $\frac{2qg_x}{3b}$, where $q\in\mathcal{Q}:=\{(1+3\epsilon)^i | i \in \mathbb{Z}, \frac{m}{1+3\epsilon} \leq (1+3\epsilon)^i \leq 2b \cdot m\}$, $g_x$ is the weight of user $x$, and $b$ is the total budget. Moreover, the modified algorithm keeps searching for a particular user who has dominated influences. The property of such a user is described by Theorem~\ref{lemmabigelement}. At the end of the modified algorithm, we might have two types of sets: one is collected by the modified threshold, and the other exists if a user described in Theorem~\ref{lemmabigelement} is found. The set with a higher objective value will be the final algorithm output. The pseudo-code is presented in Algorithm~\ref{str_b} below. Such an algorithm can solve problem~(\ref{overallproblem}) with ($\frac{1}{3}-\epsilon$)-approximation to the optimal solution according to Theorem 1 in~\cite{19}

\begin{theorem}
\label{lemmabigelement}
Assume $q\in [(1-3\epsilon)\text{OPT}, \text{OPT}$], $x$ satisfies $g_x \geq \frac{b}{2}$, and its marginal gain per weight is larger than $\frac{2q}{3b}$. Then, we have $\sigma_{mCD}(x) \geq \left(\frac{1}{3}-\epsilon\right)\text{OPT}$.
\end{theorem}

\begin{breakablealgorithm}
\caption{BUDEGETED\_STREAMING\_ALGORITHM($b,UC$)}
\label{str_b}
\begin{algorithmic}[1]
\State \textbf{Initialize} $SC[a][u] := 0$ for all actions and users; $m = 0$.
\State $\mathcal{Q}:=\{(1+3\epsilon)^i | i \in \mathbb{Z}\}$.
\State $\mathcal{S}_q := \emptyset, UC_q := UC$ and $SC_q := SC$ for all $q \in \mathcal{Q}$. 
\NoDoFor{each $x \in \mathcal{V}$}
	\State{
	 $m:= \max\{m, $MG[$x$]:= \\
	  \qquad \quad COMPUTE\_MAGINAL\_GAIN($x,UC,SC$)$/g_x\}$
	 }
	\State $\mathcal{Q}:=\{(1+3\epsilon)^i | i \in \mathbb{Z}, \frac{m}{1+3\epsilon} \leq (1+3\epsilon)^i \leq 2b \cdot m\}$.
	\NoDoFor{$q \in \mathcal{Q}$}
		\NoThenIf{$w_x \geq \frac{b}{2}$ and $\frac{MG[x]}{w_x} \geq \frac{2q}{3b}$}
			\State $\mathcal{S}_q := \{x\}$.
			\State \textbf{break}.
		\EndIf
		\NoThenIf{COMPUTE\_MAGINAL\_GAIN($x,UC_q,SC_q$) $\geq \frac{2qg_x}{3b}$ and $g^TI_{\mathcal{S}_q \cup \{x\}}\leq b$}
			\State $\mathcal{S}_q := \mathcal{S}_q \cup \{x\}$.
			\State $<UC_q, SC_q> := UPDATE(x, UC_q, SC_q)$
		\EndIf
	\EndFor
\EndFor
\State \textbf{return} $\mathcal{S} := \text{argmax}_{q \in \mathcal{Q}}\sigma_{mCD}(\mathcal{S}_q)$.
\end{algorithmic}
\end{breakablealgorithm}


\section{Experimental Results}
We conduct our experiments on a reduced Twitter dataset~\cite{15}, containing about 17,000 users and 100 actions to evaluate the mCD model and the corresponding streaming algorithms. Specifically, we are interested in the following performance metrics: 1) the influence ability of the seed set provided by our proposed streaming algorithm;  2) the gap between the output influence ability and the number of people that truly get influenced; and 3) the running time of the algorithm. All experiments are conducted at a server with a 3.50GHz Quad-Core Intel Xeon CPU E3-1245 and 32GB memory. 

\subsection{Experiment Setup}
The Twitter dataset records three different user activities, namely ``retweet", ``quote" and ``reply". In our experiments, an action $a_i$ is claimed if any user reacts (including retweet, quote, and reply to) the post of specific user $u_i$. For example, suppose there are two users, $u_1$ and $u_2$. Then, the action space could be $\mathcal{A} = \{a_1, a_2\}$. When $u_1$ performs action $a_2$, it means that user $u_1$ either ``retweets", ``quotes" or ``replies" to the twitter post of user $u_2$. Note that the cardinality of the action space may not match the cardinality of the user space. In particular, if we can only find records indicating $u_1$ either ``retweets", ``quotes" or ``replies" to the twitter post of user $u_2$, the action space is $\mathcal{A} = \{a_2\}$, where $|\mathcal{A}|=1$. In this way, the considered dataset contains 17,000 users and about 100 different actions. According to the discussion in Section~3, the event log is divided into two parts, where the training set contains 80 different actions and the test set contains 20 different actions.

\subsection{Experimental Results}
In this subsection, we are going to show:
\begin{itemize}
   \item The seed set identified by the mCD model has better quality in the sense of the influence ability.  
   \item Under the same mCD model, the streaming algorithms can achieve close utilities to the Cost-Effective Lazy Forward selection (CELF) algorithm proposed in~\cite{10}, an accelerated greedy algorithm;
   \item Under the same mCD model, the streaming algorithms are much faster than the CELF algorithm.
  \item As an estimator on the total number of users that get influenced, the mCD model is more accurate than the conventional CD model; 
  
\end{itemize}

Note that it has been about ten years after CELF was proposed, and there have been more recent influence maximization algorithms proposed in, e.g., \cite{martingales} and \cite{mssa}. However, these recent ones were designed to apply different influence models and to solve differently formulated problems, compared with ours. Thus, no performance comparison is given against these methods in this paper as it may not be a fair comparison under our setup. To be more specific, note that one of the main contributions in our paper is to study the budgeted influence maximization problem. Therefore, the experiments mainly focus on how to select a group of users to maximize the influence considering the selection cost of users. However, none of the above two papers studied such a budgeted influence maximization problem. In addition, the algorithms proposed in papers \cite{martingales} and \cite{mssa} are applicable to influence maximization problem under the diffusion models, while our streaming algorithm is built upon the credit distribution model, which is under a quite different problem setup.

For notation simplicity, the output results of the CELF algorithm and the streaming algorithm under the mCD model are denoted by ``mCD\_greedy" and ``mCD\_streaming" respectively. 

\subsubsection{Quality of Seed Sets}
First, we focus on the evaluation of the following three models: 
\begin{itemize}
\item \textbf{IC model}: a conventional non-credit distribution model with edge probabilities assigned as $0.1$ uniformly (in all experiments, we run $10,000$ MC simulations)~\cite{5}; 
\item \textbf{CD model}: with direct credit assigned as described in~\cite{12} and the CELF algorithm is used to produce solutions.
\item \textbf{mCD model}: a multi-action credit distribution model proposed in this paper and the CELF algorithm is used to produce solutions.
\end{itemize}

After the seed sets produced by these three kinds of approaches, we compare the influence ability of different results on the mCD model to verify the quality of seed sets. It can be observed in Fig.~\ref{fig:quality} that the influence ability of seed sets picked by the mCD model is larger than those running on the IC models and the conventional CD model. For instance, when $k=50$, the influence ability of the seed set picked by the CELF algorithm on the mCD model is 1350.08, while the influence ability on other two models (CD, IC) is 1280.24 and 220.03, respectively. Based on the curve in Fig.~\ref{fig:quality}, we conclude that our proposed model has an improved capability in identifying seed sets and describing the influence propagation in online social networks. 
\begin{figure}
   \centering
   \includegraphics[width=0.8\linewidth]{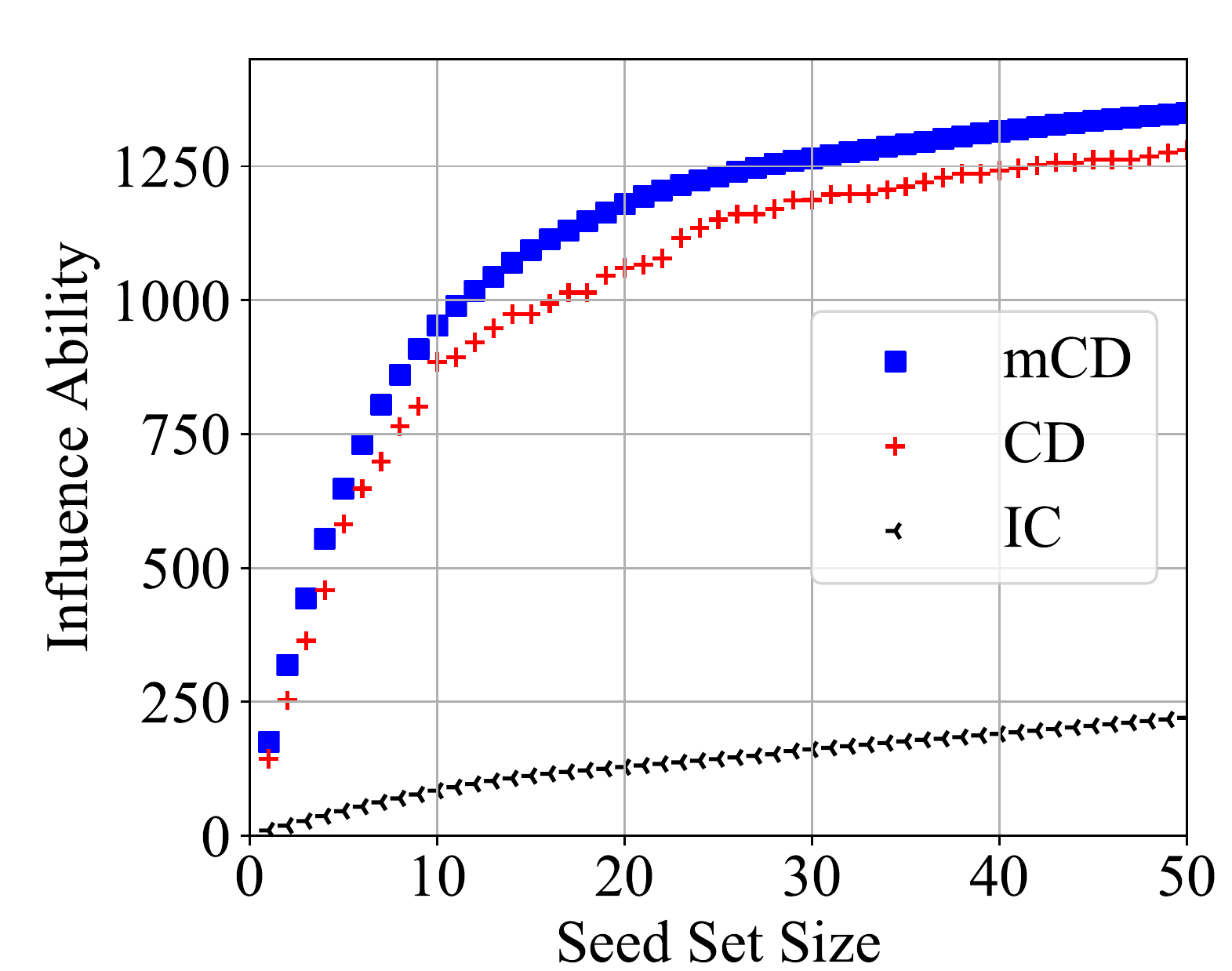} 
   \caption{Influence Ability Comparison over Different Models.}
   \label{fig:quality}
\end{figure}

\subsubsection{Influence Ability of Seed Set}
Next, we compare the influence ability of different seed sets obtained by our proposed streaming algorithms and the CELF algorithm under the same mCD model. For both the influence maximization problem and budgeted influence maximization problem, from Fig.~\ref{fig: quality} and~\ref{fig:b_quality}, we can observe that the seed sets provided by our proposed streaming algorithms can achieve utilities close to the CELF algorithm. For instance, in Fig.~\ref{fig: quality}, when $k = 50$, a seed set with 1333.56 influence ability is given by the streaming algorithm (Algorithm~\ref{str_k}), which is only 1\% less than the influence ability given by the CELF algorithm. Moreover, in Fig.~\ref{fig:b_quality}, taking $b=500$ as an example, the influence ability of the seed set provided by the streaming algorithm (Algorithm~\ref{str_b}) is 0.1\% less than the CELF algorithm. Therefore, we conclude that our proposed streaming algorithms are sufficient to identify seed sets with close influence ability to the CELF algorithm.
\begin{figure}
   \centering
   \includegraphics[width = 0.9\linewidth]{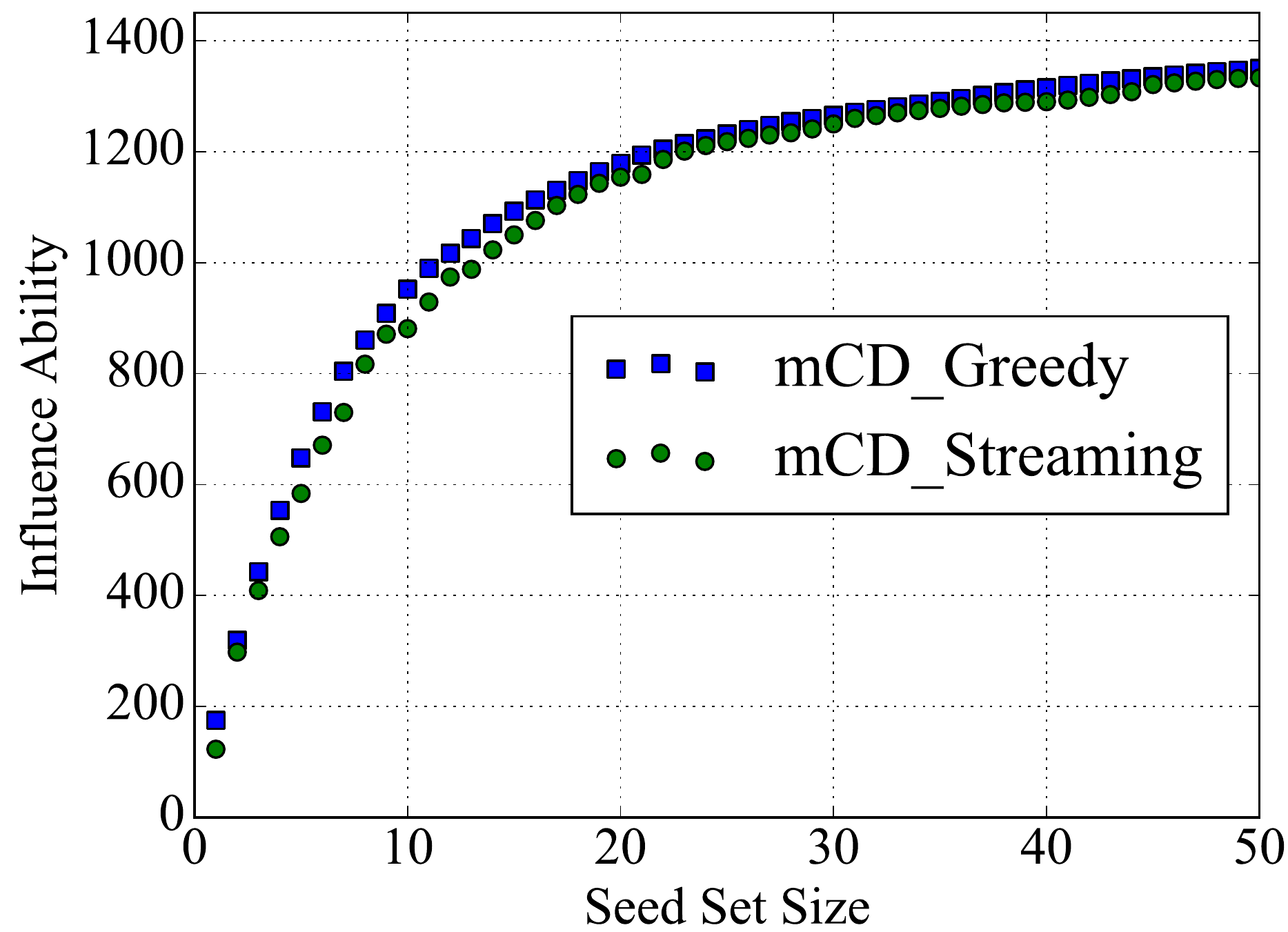}
   \caption{Influence Ability Comparison under mCD Model with the Cardinality Constraint.}
   \label{fig: quality}
\end{figure} 
\begin{figure} 
   \centering
   \includegraphics[width=0.85\linewidth]{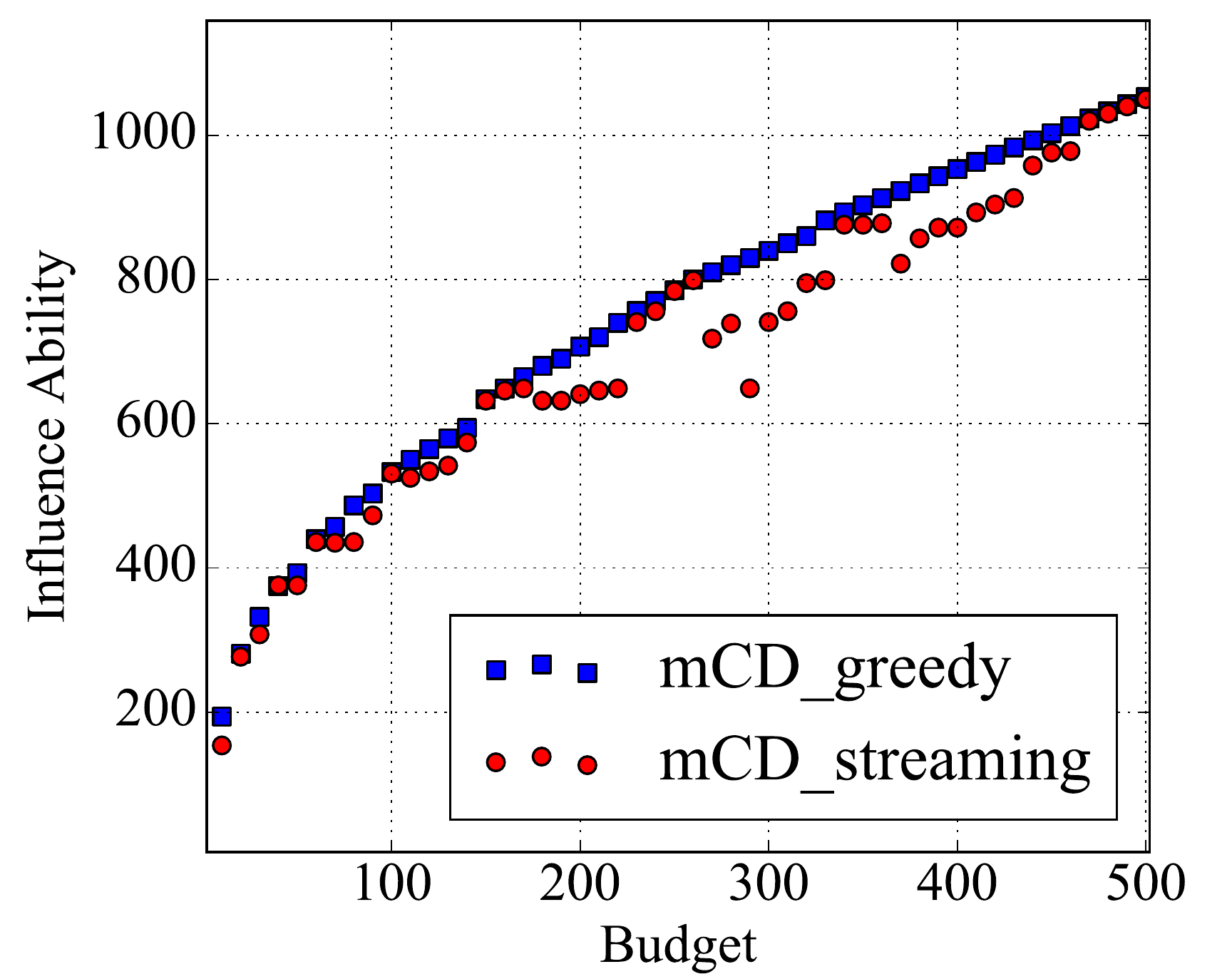} 
   \caption{Influence Ability Comparison under mCD Model with the Budget Constraint.}
   \label{fig:b_quality}
\end{figure}

\subsubsection{Algorithm Running Time}
It has been shown in Fig. \ref{fig: quality} and Fig.~\ref{fig:b_quality} that our proposed streaming algorithms are able to provide performance close to the CELF algorithm. Such an algorithm maintains a list of current marginal gains of all the elements in the ground set, keeps them updated, and sorts the list in an increasing order recursively. Unlike the CELF algorithm, the proposed streaming algorithm only requires one scan over the user set. Therefore, the resulting lower computation complexity makes the streaming algorithm more practical when the number of elements in the ground set is large. To further examine this argument, we compare the running time (in seconds) for the CELF and the streaming algorithm with the same mCD model in Figs.~\ref{fig:r} and \ref{fig:rb}. It can be seen that for both the influence maximization and budgeted influence maximization problem, our proposed streaming algorithm is several orders of magnitude faster. Especially, in the case of the budgeted influence maximization problem, when the budget is set to be $500$, it takes more than $3,800$ seconds to complete the whole process in CELF, while for the streaming algorithm, it only takes $5.3$ seconds. Meanwhile, the streaming algorithm achieves almost the same performance as CELF, which implies that our proposed streaming algorithm is both efficient and effective. 
\begin{figure}
   \centering
   \includegraphics[width = 1\linewidth]{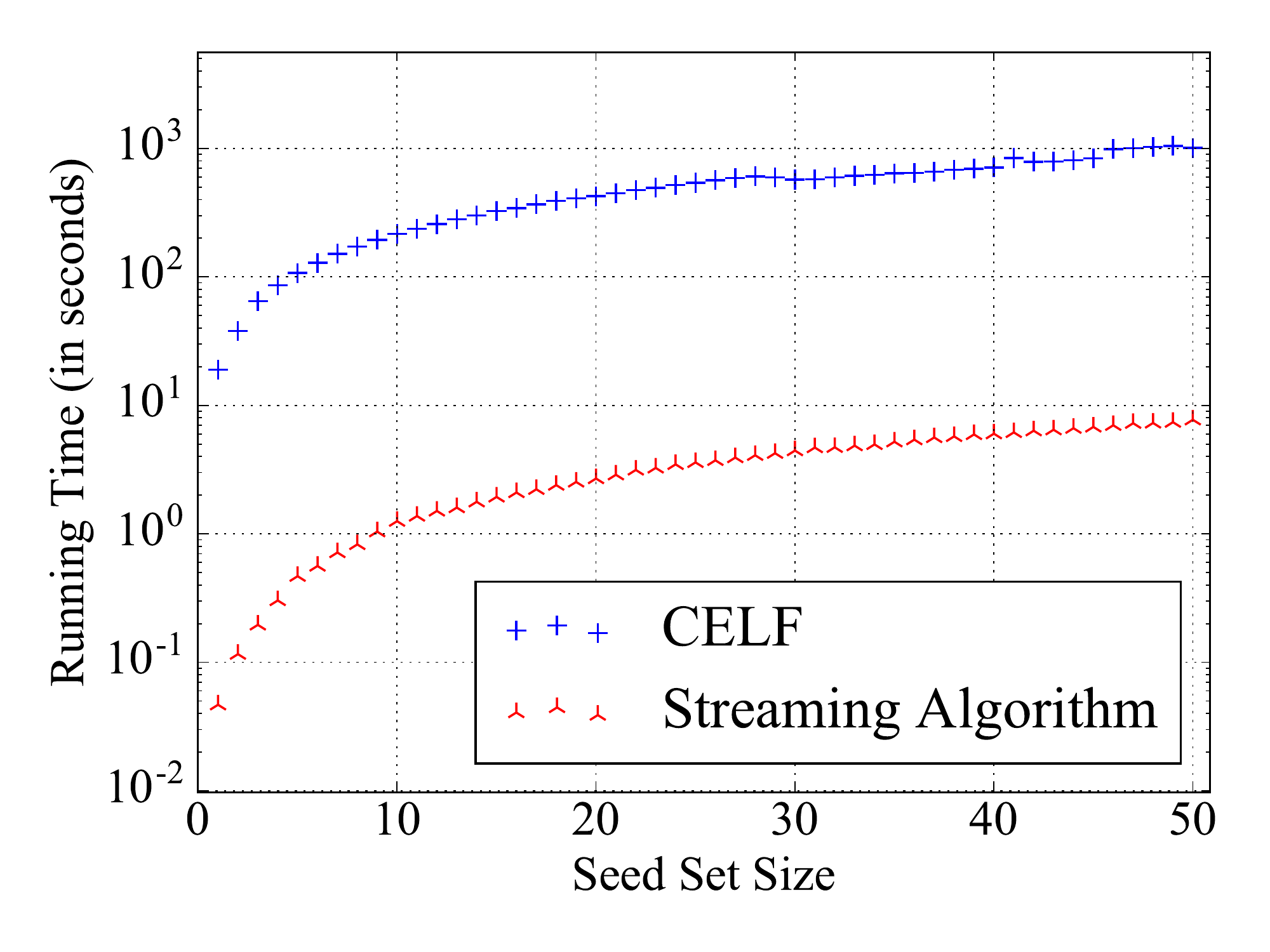} 
   \caption{Running Time Comparison with the Cardinality Constraint.}
   \label{fig:r}
\end{figure}
\begin{figure}
   \centering
   \includegraphics[width = 0.9\linewidth]{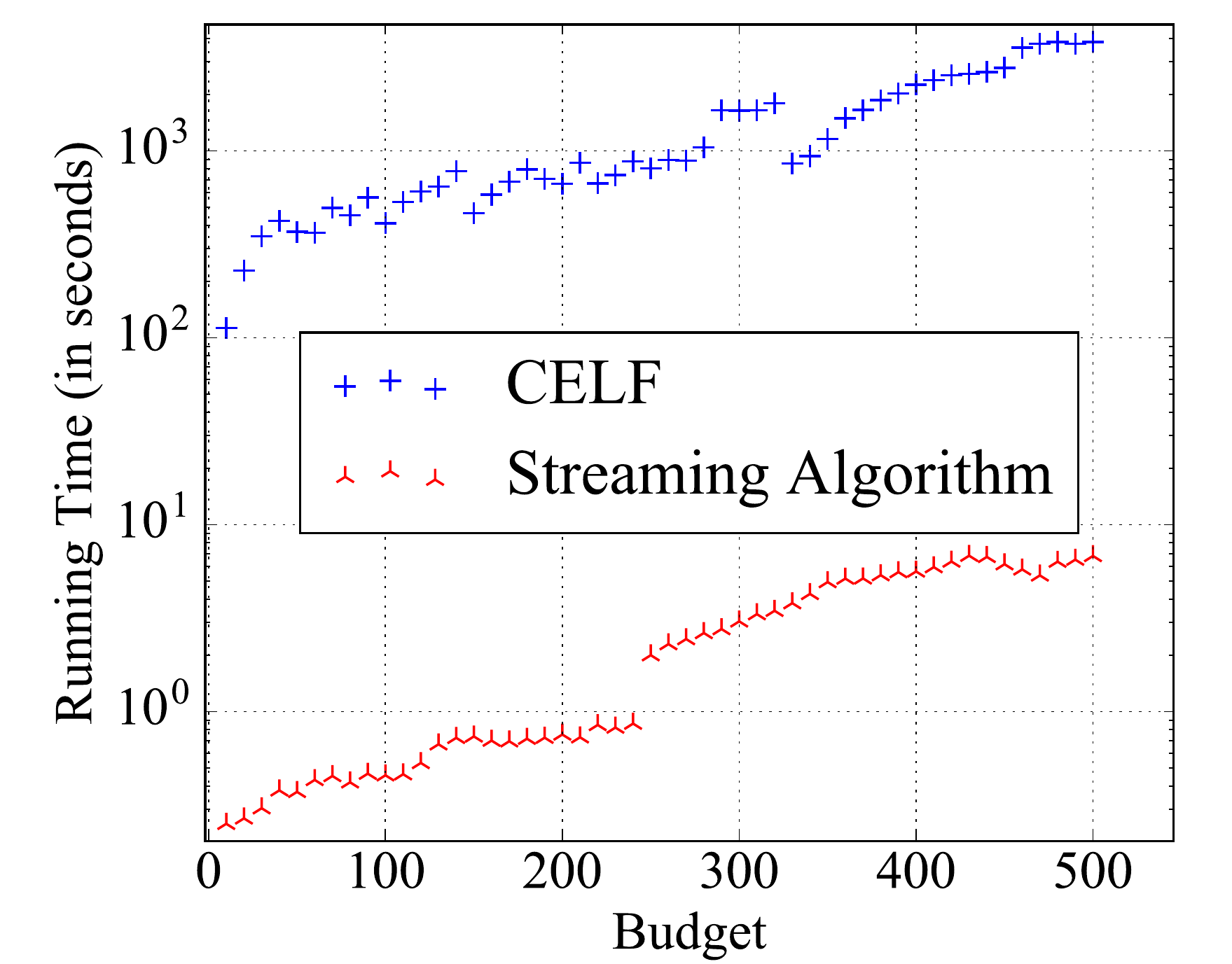} 
   \caption{Running Time Comparison with the Budget Constraint.}
   \label{fig:rb}
\end{figure}

\subsubsection{Estimation on the Number of Influenced People}
Our goal in this experiment is to investigate how the mCD model performs in estimating the number of people that get influenced in the network. We pick 950 actions from the original dataset for this experiment. Since the streaming algorithm (Algorithm \ref{str_k}) is able to achieve close performance to CELF with a much faster speed, we only conduct the streaming algorithm to explore the estimation accuracy. Note that when we set the seed set size equal to the number of initiators for a particular action, the mCD model can always provide the actual number of influenced people \footnote{Given an action, the ground truth is always accessible by simply counting the number of people performing the give action in the dataset.}. Then, we fix the size of the seed set as 50 to explore the estimation accuracy. To better illustrate, we sort actions with increasing popularity. It can be observed in Fig.~\ref{fig:test} that the estimated values obtained by both the CD and the mCD models are smaller than the actual number of users performing the corresponding action in the social network. However, it can be seen that the estimated results by our proposed model are closer to the true values, which means that the estimation with our model is more accurate for a given seed set size.
\begin{figure}
   \centering
   \includegraphics[width = 0.8\linewidth]{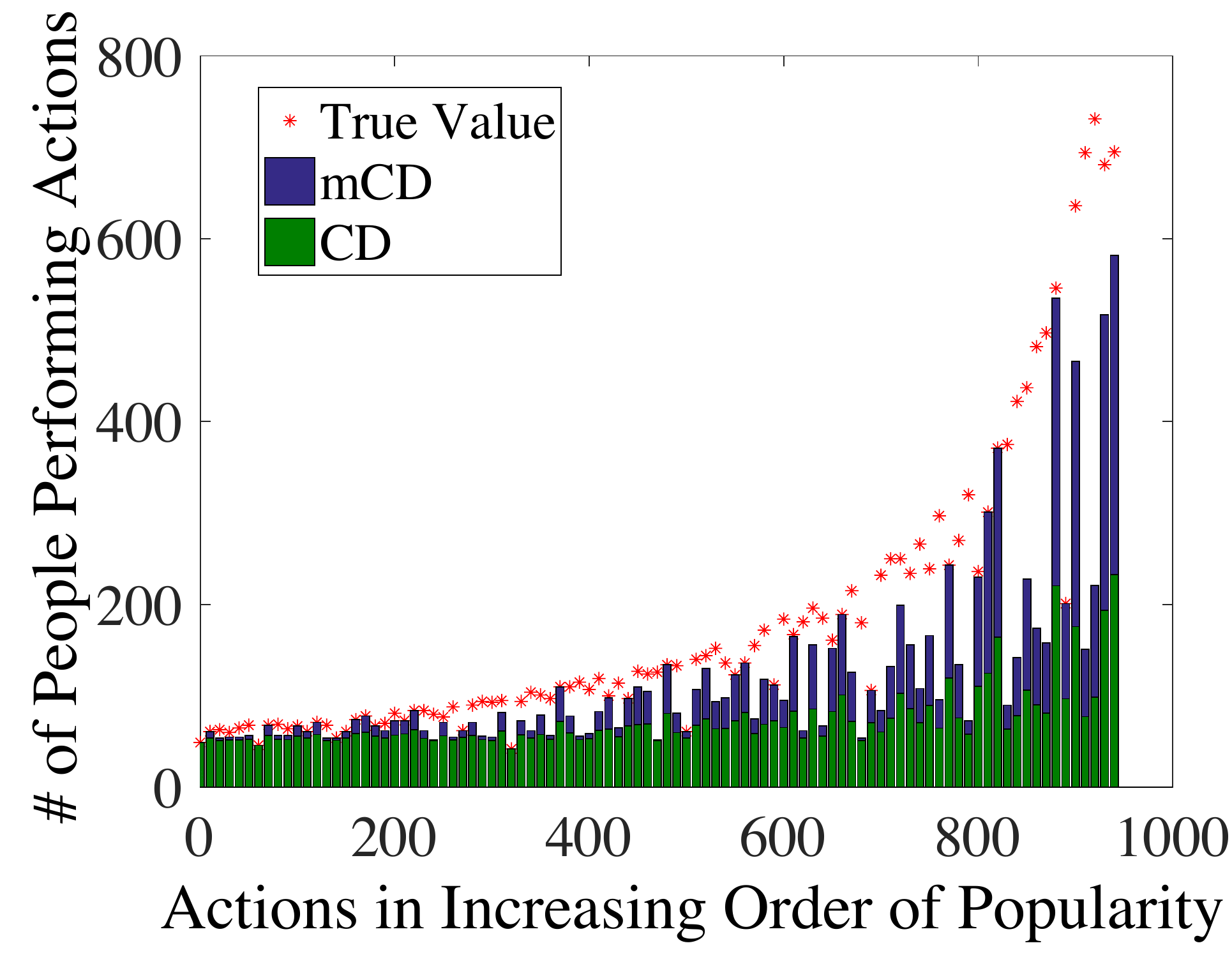} 
   \caption{Estimated Influence for Actions in Test Set.}
   \label{fig:test}
\end{figure}

\section{Conclusion} 
Our work is novel in three folds: 1) we are the first to study the multi-action event log by extending the credit distribution model, which cannot be directly derived from Goyal's work in \cite{12}; 2) different from previous papers, we focus on the budgeted influence maximization problem under credit distribution models, instead of the influence maximization problem under propagation models that involve edge weights; and 3) we propose a streaming algorithm to solve the budgeted influence maximization problem, whose theoretical analysis is different from Badanidiyuru's results in \cite{18}.

More specifically, we extended the conventional CD model to the mCD model in dealing with the multi-action event log and analyzing the influence ability of users in online social networks. Based on this credit model, an efficient streaming algorithm was developed to provide a solution with ($\frac{1}{2}-\epsilon$)-approximation of the optimal value under the influence maximization problem with a cardinality constraint, and ($\frac{1}{3}-\epsilon$)-approximation under the budgeted influence maximization problem. More specifically, we re-designed the credit assignment method in the CD model by utilizing a modified harmonic mean to handle multi-action event logs. This new credit assignment method not only makes full use of the multi-action event log but also achieves higher accuracy in estimating the total number of people that get influenced without the edge weight assignment and expensive Monte-Carlo simulations. Experiments showed that the mCD model is more accurate than the conventional CD model, and able to identify a seed set with higher quality than both the IC and CD models. Even under the same mCD model, the proposed streaming algorithms are able to achieve similar performance to the CELF greedy algorithm, but several orders of magnitude faster.

\end{document}